\documentclass[11pt]{article}

\usepackage{amsmath, amssymb, amsthm}
\usepackage{graphicx}
\usepackage{comment}

\newcommand*{\tw}{\textsc{2}}
\newcommand*{\twdfa}{\tw\textsc{dfa}}
\newcommand*{\twdfas}{\tw\textsc{dfa}s}

\newcommand*{\ow}{\textsc{1}}
\newcommand*{\owdfa}{\ow\textsc{dfa}}
\newcommand*{\owdfas}{\ow\textsc{dfa}s}
\newcommand*{\ownfa}{\ow\textsc{nfa}}
\newcommand*{\ownfas}{\ow\textsc{nfa}s}

\newcommand*{\cfg}{\textsc{cfg}}
\newcommand*{\cfgs}{\textsc{cfg}s}
\newcommand*{\cnfg}{\textsc{Cnfg}}
\newcommand*{\cnfgs}{\textsc{Cnfg}s}

\newcommand*{\cfls}{\textsc{cfl}s}

\newcommand{\derivesForm}[2]{\mathrel{\underset{\scriptstyle #1}{\overset{\scriptstyle #2}{\Longrightarrow}}}}

\newcommand{\DerivesGrammar}[1]{\derivesForm{#1}{+}}
\newcommand{\DERIVESGrammar}[2]{\derivesForm{#1}{#2}}
\newcommand{\derGrammar}[1]{\DERIVESGrammar{#1}{}} 
 
\newcommand{\DerivesG}{\DerivesGrammar{G}} 

\newcommand{\derG}{\derGrammar{G}} 
 
\newcommand{\DerivesZ}{\DerivesGrammar{G_0}} 

\newcommand{\derZ}{\derGrammar{G_0}} 
\newcommand{\Pred}{\ensuremath{\mathrm{Pred}}}

\newcommand{\Pref}{\ensuremath{\mathrm{Pref}}}

\def\Vec#1{\mbox{\boldmath $#1$}}
\newtheorem{theorem}{Theorem}[section]
\newtheorem{lemma}[theorem]{Lemma}

\newtheorem*{claim}{Claim}     %requires amsthm package
\newtheorem{ex}[theorem]{Example}
\newenvironment{example}{\begin{ex}\rm}{\end{ex}}
\begin{document}%
%\begin{frontmatter}%
%\journal{Information and Computation}%
\title{Converting Nondeterministic Automata and Context-Free Grammars into Parikh Equivalent One-Way and
Two-Way Deterministic Automata\footnote{Preliminary version presented at
  \emph{DLT'12---\,Developments in Language Theory}, Taipei,
  Taiwan, Aug 14--17, 2012 [\emph{Lect.\ Notes Comput.\
  Sci.},\ 7410, pp.\ 284--295, Springer-Verlag, 2012]\@.}%
}%
\author{Giovanna J. Lavado\footnotemark[3]\and
Giovanni Pighizzini\footnotemark[3]\and
Shinnosuke Seki\footnotemark[4]\and
\mbox{}\\
{\normalsize \mbox{\footnotemark[3] }\,Dipartimento di Informatica}\\
{\normalsize  Universit\`{a} degli Studi di Milano, Italy}\\
{\small\sf giovanna.lavado@unimi.it} --
{\small\sf giovanni.pighizzini@unimi.it}\\[2ex]%
{\normalsize \mbox{\footnotemark[4] }\,Department of Information and Computer Science}\\
{\normalsize Aalto University, Finland}\\
{\small\sf shinnosuke.seki@aalto.fi}%[
}
%\ead{giovanni.pighizzini@unimi.it}%
%\ead{giovanna.lavado@unimi.it}%
%\ead{shinnosuke.seki@aalto.fi}%
%\address[gio]{Dipartimento di Informatica,
%  Universit\`{a} degli Studi di Milano, Italy}%
%\address[shi]{Department of Information and Computer Science,
%	Aalto University, Finland}%
\date{}%
\maketitle\thispagestyle{empty}%%
\maketitle\thispagestyle{empty}%%
\begin{quotation}\small\noindent
  \textbf{Abstract.}\hspace{\labelsep}%
We investigate the conversion of one-way nondeterministic finite automata and context-free grammars into
Parikh equivalent one-way and two-way deterministic finite automata, {}from a descriptional complexity point of view.

We prove that for each one-way nondeterministic automaton with $n$ states there exist Parikh equivalent
one-way and two-way deterministic automata with $e^{O(\sqrt{n \cdot \ln n})}$ and~$p(n)$ states, respectively,
where~$p(n)$ is a polynomial. Furthermore, these costs are tight.
In contrast, if all the words accepted by the given automaton contain at least two different
letters, then a Parikh equivalent one-way deterministic automaton with a polynomial number of states can
be found.

Concerning context-free grammars, we prove that for each grammar in Chomsky normal form with $h$ 
variables there exist Parikh equivalent one-way and two-way deterministic automata with~$2^{O(h^2)}$ 
and~$2^{O(h)}$ states, respectively.
Even these bounds are tight.\\
\mbox{}\\
\textbf{Keywords:}\hspace{\labelsep}%
  finite automaton,
	context-free grammar, 
	Parikh's theorem,
	descriptional complexity,
	semilinear set, 
	Parikh equivalence
\end{quotation}

%----------------------------------------------------------------
	\section{Introduction}
%----------------------------------------------------------------

It is well-known that the state cost of the conversion of nondeterministic finite automata (\ownfas) into
equivalent deterministic finite automata (\owdfas) is exponential:
using the classical subset construction~\cite{RabinScott1959}, {}from each $n$-state \ownfa\ we can build 
an equivalent \owdfa\ with $2^n$ states.
Furthermore, this cost cannot be reduced.

In all examples witnessing such a state gap (e.g.,~\cite{Lupanov1963,MeyerFischer1971,Moore1971}),  
input alphabets with at least two letters
and proof arguments strongly relying on the structure of words are used.
As a matter of fact, for the unary case, namely the case of the one letter input alphabet, the cost
reduces to $e^{\Theta(\sqrt{n \cdot \ln n})}$, as shown by Chrobak~\cite{Chrobak1986}.

\begin{quote}
\emph{What happens if we do not care of the order of symbols in the words, i.e., if we are interested
only in obtaining \owdfas\ accepting sets of words which are equal, after permuting
the symbols, to the words accepted by the given \ownfas?}
\end{quote}

This question is related to the well-known notions of Parikh image and Parikh equivalence~\cite{Parikh1966},
which have been extensively investigated in the literature
(e.g., \cite{Goldstine77,AcetoEzikIngo02}) even for the connections of semilinear
sets~\cite{Huynh80} and with other fields of investigation as, e.g., 
Presburger Arithmetics~\cite{GinsburgSpanier66}, Petri Nets~\cite{Esparza97}, logical formulas~\cite{VermaSeidlSchwentick05},
formal verification~\cite{To10b}.

We remind the reader that two words over a same alphabet $\Sigma$ are Parikh equivalent if and only if they are equal up to a
permutation of their symbols or, equivalently, for each letter $a\in\Sigma$, the number of occurrences
of $a$ in the two words is the same.
This notion extends in a natural way to languages (two languages $L_1$ and $L_2$ are Parikh equivalent when
for each word in $L_1$ there is a Parikh equivalent word in $L_2$ and vice versa) and to formal
systems which are used to specify languages as, for instance, grammars and automata.
Notice that in the unary case Parikh equivalence is just the standard equivalence. So, in the unary case,
the answer to our previous question  is given by the above mentioned result by Chrobak.

Our first contribution in this paper is an answer to that question in the general case.
In particular, we prove that the state cost of the conversion of $n$-state \ownfas\ into Parikh equivalent \owdfas\ is the
same as in the unary case, i.e., it is $e^{\Theta(\sqrt{n \cdot \ln n})}$.
More surprisingly, we prove that this is due to the unary parts of languages.
In fact, we show that if the given \ownfa\ accepts only nonunary words, i.e., each accepted word contains at
least two different letters, then we can obtain a Parikh equivalent \owdfa\ with a 
polynomial number of states in $n$.
Hence, while in standard determinization the most difficult part (with respect to the state complexity)
is the nonunary one, in the ``Parikh determinization'' this part becomes easy and the most complex part is
the unary one.

In the second part of the paper we consider context-free grammars (\cfgs).
Parikh Theorem~\cite{Parikh1966} states that each context-free language is Parikh equivalent to a
regular language.
We study this equivalence {}from a descriptional complexity point of view.\footnote{For an introductory
survey to \emph{descriptional complexity}, we address the reader to~\cite{GoldstineKappesKLMW2002}.}
Recently, Esparza, Ganty, Kiefer, and Luttenberger proved that each \cfg~$G$ in Chomsky normal
form with $h$ variables can be converted into a Parikh equivalent \ownfa\ with $O(4^h)$ 
states~\cite{EsGaKiLu2011}.
In~\cite{LavadoPighizzini2012} it was proved that if $G$ generates a bounded language then we can 
obtain a \owdfa\ with $2^{h^{O(1)}}$ states, i.e., a number exponential in a polynomial of the number of
variables.
In this paper, we are able to extend such a result by removing the restriction to bounded languages.
We also reduce the upper bound to $2^{O(h^2)}$.
A milestone for obtaining this result is the conversion of \ownfas\ to Parikh equivalent \owdfas\
presented in the first part of the paper. By suitably combining that conversion (in particular the polynomial conversion in the case of \ownfas\ accepting nonunary words) with the above mentioned result 
{}from~\cite{EsGaKiLu2011} and with a result by Pighizzini, Shallit, and Wang~\cite{PighizziniShallitWang2002} concerning
the unary case, we prove that each context-free grammar in Chomsky normal form
with $h$ variables can be converted into
a Parikh equivalent \owdfa\ with $2^{O(h^2)}$ states. {}From the results concerning the unary case, it
follows that this bound is tight.

Even for this simulation, as for that of \ownfas\ by Parikh equivalent \owdfas, the main contribution to
the state complexity of the resulting automaton is given by the unary  part.

Finally, we consider conversions of \ownfas\ and \cfgs\ into Parikh equivalent \emph{two-way deterministic
automata}~(\twdfas). Due to the fact that in the unary case these conversions are less expensive than the corresponding
ones into \owdfas, we are able to prove that each $n$-state  \ownfa\ can be converted into an equivalent \twdfa\
with a number of states polynomial in~$n$, and each context-free grammar in Chomsky normal form with~$h$ variables
can be converted into a Parikh equivalent \twdfa\ with a number of states exponential in~$h$.

%----------------------------------------------------------------
	\section{Preliminaries}
	\label{sec:preliminaries}
%----------------------------------------------------------------

We assume the readers to be familiar with the basic notions and properties from automata and formal language
theory. We remind just a few notions, addressing the reader to standard textbooks (e.g., \cite{HopcroftUllman1979,Shallit2008}) 
for further details.

\smallskip

Let $\Sigma = \{a_1, a_2, \ldots, a_m\}$ be an alphabet of $m$ letters. 
Let us denote by~$\Sigma^*$ the set of all words over $\Sigma$ including the \emph{empty word} $\varepsilon$. 
Given a word $w \in \Sigma^*$, $|w|$ denotes its \emph{length} and, for a letter $a \in \Sigma$, $|w|_a$ denotes 
the \emph{number of occurrences} of $a$ in $w$. 
For a word $u \in \Sigma^*$, $w$ is a \emph{prefix} of $u$ if $u = wx$ for some word $x \in \Sigma^*$.
If $x\neq\epsilon$ then $w$ is a \emph{proper prefix} of $u$.
We denote by $\Pref(u)$ the set of all prefixes of $u$, and for a language $L \subseteq \Sigma^*$, let
\[
\Pref(L) = \bigcup_{u \in L} \Pref(u)\,.
\] 
A language $L$ has the \emph{prefix property} or, equivalently, is said to be \emph{prefix-free} if and only if
for each word $x\in L$, each proper prefix of $x$ does not belong to~$L$.
Given two sets $A,B$ and a function $f:A\rightarrow\Sigma^*$, we say that $f$ has the prefix property on $B$ if
and only if the language $f(A\cap B)$ has the prefix property.

\smallskip

In the paper we consider:
\begin{itemize}
\item \emph{one-way deterministic} and \emph{one-way nondeterministic} finite automata (abbreviated as \owdfas\ and \ownfas, respectively),
\item \emph{two-way deterministic} finite automata (\twdfas),
\item \emph{context-free grammars} (\cfgs) and \emph{context-free languages} (\cfls).
\end{itemize}
While in \emph{one-way automata} the input is scanned from left to right, 
until reaching the end of the input, where the word is accepted or rejected,
in \emph{two-way automata}\label{p:twdfa} the head can be moved in both directions.
At each step, depending on the current state and scanned input symbol and according to the transition function, the 
internal state is changed and 
the head is moved one position leftward, one position rightward or it is kept on the same
cell. In order to locate
the left and the right ends of the input, the word is given on the tape surrounded by two special symbols,
the \emph{left and right endmarkers}. We assume that a \twdfa\ starts the computation in a designed initial state,
scanning the first input symbol and that its head cannot violate the endmarkes, namely, there are no transitions
reading the left (right) endmarker and moving to the left (right, respectively).
In the literature, several slightly different acceptance conditions for two-way automata have been considered.
Here, we assume that a \twdfa\ accepts the input by entering in a special state $q_f$ which is also halting.
We use~$L(A)$ to denote the \emph{language accepted} or \emph{defined} by an automaton~$A$.

\smallskip

A \cfg\ $G$ is denoted by a quadruple $(V, \Sigma, P, S)$, where $V$ is the set of variables, 
$\Sigma$ is the terminal alphabet, $P$ is the set of productions, 
and $S \in V$ is the start variable. 
By $L(G)$ we denote the \emph{language generated} or \emph{defined} by~$G$, namely the set of all words in $\Sigma^*$ that have at least one derivation by $G$ {}from $S$.  
$G$ is said to be in \emph{Chomsky normal form} if all of its productions are in one of the three simple forms, 
either $B \to CD$, $B \to a$, or $S \to \varepsilon$, where $a \in \Sigma$, $B \in V$, and $C, D \in V \setminus \{S\}$. 
\cfgs\ in Chomsky normal form are called {\it Chomsky normal form grammars} (\cnfgs). 
According to the discussion in~\cite{Gruska1973}, we employ the number of variables of \cnfgs\ as a ``reasonable'' measure of descriptional complexity for~\cfls. 

\medskip
A word is said to be \emph{unary} if it consists of $k\geq 0$ occurrences of a same symbol,
otherwise it is said to be \emph{nonunary}.
A language $L$ is \emph{unary} if $L\subseteq\{a\}^*$ for some letter $a$.
In a similar way, automata and \cfgs\ are \emph{unary} when their input and terminal alphabets, respectively, consist of
just one symbol.

Given an alphabet $\Sigma=\{a_1,a_2,\ldots,a_m\}$ and a language $L\subseteq\Sigma^*$,
the \emph{unary parts} of $L$ are the  languages
\[
L_1=L\cap\{a_1\}^*,~L_2=L\cap\{a_2\}^*,\ldots,~L_m=L\cap\{a_m\}^*\,,
\]
while the \emph{nonunary part} is the language
\[
L_0=L-\left(L_1\cup L_2\cup\ldots\cup L_m\right)\,,
\]
i.e., the set which consists of all nonunary words belonging to the language~$L$.
Clearly,~$L=\bigcup_{i=0}^{m}L_i$.

\medskip

We denote the set of integers by~$\mathbb{Z}$ and the set of nonnegative integers by~$\mathbb{N}$.
Then $\mathbb{Z}^m$ and $\mathbb{N}^m$ denote the corresponding sets of $m$-dimensional integer vectors 
including the {\it null vector} $\Vec{0} = (0, 0, \ldots, 0)$. 
For $1 \le i \le m$, we denote the $i$-th component of a vector $\Vec{v}$ by $\Vec{v}[i]$. 

Given $k$ vectors $\Vec{v}_1,\ldots,\Vec{v}_k\in\mathbb{Z}^m$, we say that they are \emph{linearly independent}
if and only if for all $n_1,\dots n_k\in\mathbb{Z}$, 
$n_1\Vec{v}_1+\cdots+n_k\Vec{v}_k=\Vec{0}$ implies $n_1=\ldots=n_k=0$.
It is well-known that, in this case, $k$ cannot exceed~$m$.
The following result will be used in the paper.

\begin{lemma}\label{lem:indep}
  Given $k$ linearly independent vectors $\Vec{v}_1,\ldots,\Vec{v}_k\in\mathbb{Z}^m$ there are $k$ pairwise
  different integers $t_1,\ldots,t_k\in\{1,\ldots,m\}$ such that $\Vec{v}_j[t_j]\neq 0$, for $j=1,\ldots,k$.
\end{lemma}
\begin{proof}
  Let $V$ be the $m\times k$ matrix which has $\Vec{v}_1,\ldots,\Vec{v}_k$ as columns. Since the given vectors
  are linearly independent, $k\le m$. 
  First, we suppose $k=m$. Then the determinant $d(V)$ of $V$ is defined and it is nonnull. 
  
  If $k=1$ then the result is trivial.
  Otherwise, we can compute $d(V)$ along the last column as
  \[
  d(V)=\sum_{i=1}^k(-1)^{i+k}\Vec{v}_k[i]d_{i,k}\,,
  \] 
  where $d_{i,k}$ is the determinant of the matrix $V_{i,k}$ obtained by
  removing from $V$ the row $i$ and the column $k$. Since $d(V)\neq 0$, there is at least one index $i$ such that
  $\Vec{v}_k[i]$ and $d_{i,k}\neq 0$. Hence, as~$t_k$ we take such~$i$.
  Using an induction on the matrix $V_{t_k,k}$, we can finally obtain the sequence $t_1,\ldots,t_k$ 
  satisfying the statement of the theorem.
  
  Finally, we observe that in the case $k<m$, by suitably deleting $m-k$ rows from $V$, we obtain
  a~$k\times k$ matrix $V'$ with $d(V')\neq 0$. Thus, we can apply the same argument to $V'$.
%\qed  
\end{proof}

A vector $\Vec{v} \in \mathbb{Z}^m$ is {\it unary} if it contains at most one nonzero component, i.e., $\Vec{v}[i], \Vec{v}[j] \neq 0$ for some $1 \le i, j \le m$ implies $i = j$; otherwise, it is {\it nonunary}. 
By definition, the null vector is unary.

In the sequel, we reserve $\preceq$ for the componentwise partial order on $\mathbb{N}^m$, i.e., $\Vec{u} \preceq \Vec{v}$ if and only if $\Vec{u}[k] \le \Vec{v}[k]$ for all $1 \le k \le m$. 
For a vector $\Vec{v} \in \mathbb{N}^m$, let 
\[
\Pred(\Vec{v}) = \{\Vec{u} \mid \Vec{u} \preceq \Vec{v}\}\,.
\] 
For $\Vec{u}, \Vec{v} \in \mathbb{N}^m$, $\Vec{v} - \Vec{u}$ is defined to be a vector $\Vec{w}$ with $\Vec{w}[k] = \Vec{v}[k] - \Vec{u}[k]$ for all $1 \le k \le m$. 
Note that $\Vec{v} - \Vec{u}$ is a vector in $\mathbb{N}^m$ if and only if $\Vec{u} \preceq \Vec{v}$. 

\medskip

A \emph{linear set} in $\mathbb{N}^m$ is a set of the form
\begin{equation}
\left\{\Vec{v}_0 + n_1\Vec{v}_1+n_2 \Vec{v}_2 +\cdots+n_k\Vec{v}_k \mid n_1, n_2,\ldots, n_k \in \mathbb{N}\right\}\,,\label{eq:linear}
\end{equation}
where $k\geq 0$ and $\Vec{v}_0, \Vec{v}_1, \Vec{v}_2, \ldots, \Vec{v}_k \in \mathbb{N}^m$. 
The vector $\Vec{v}_0$ is called \emph{offset}, while the vectors $\Vec{v}_1, \ldots, \Vec{v}_k$
are called \emph{generators}.
A \emph{semilinear set} in $\mathbb{N}^m$ is a finite union of linear sets in $\mathbb{N}^m$.

\medskip

The {\it Parikh map} $\psi: \Sigma^* \to \mathbb{N}^m$ associates with a word $w \in \Sigma^*$ the vector
\[
\psi(w)=\left(|w|_{a_1}, |w|_{a_2}, \ldots, |w|_{a_m}\right)\,, 
\]
which counts the occurrences of each letter of $\Sigma$ in~$w$. The vector~$\psi(w)$ is also called \emph{Parikh image} of $w$.
Notice that a word $w\in\Sigma^*$ is {\it unary} if and only if its Parikh image $\psi(w)$ is a unary vector.
One can naturally generalize this map for a language $L \subseteq \Sigma^*$ as
\[
\psi(L) = \{\psi(w) \mid w \in L\}\,.
\]
The set~$\psi(L)$ is called the \emph{Parikh image} of $L$. 
Two languages $L, L' \subseteq \Sigma^*$ are said to be {\it Parikh equivalent} to each other if and only if $\psi(L) = \psi(L')$.

\medskip

Parikh equivalence can be defined not only between languages but among languages, grammars, and finite automata by referring,
in the last two cases, to the defined languages.
For example, given a language $L$, a \cfg\ $G$, and a finite automata $A$, we say that:
\begin{itemize}
\item $G$ is Parikh equivalent to $L$ if and only if $\psi(L(G)) = \psi(L)$,
\item $A$ is Parikh equivalent to $L$ if and only if $\psi(L(A)) = \psi(L)$,
\item $G$ is Parikh equivalent to $A$ if and only if $\psi(L(G)) = \psi(L(A))$.
\end{itemize}
Parikh's Theorem, proven in 1966~\cite{Parikh1966}, states that \emph{the Parikh image of any context-free language is a semilinear set.} 
Since the class of regular languages is closed under union and each linear set as in~(\ref{eq:linear}) is the Parikh image of 
the regular language
\[
\{w_0\}\cdot\{w_1,w_2,\ldots,w_k\}^*\,,
\]
where, for $i=0,\ldots,k$,
$
w_i = a_1^{\Vec{v}_i[1]}a_2^{\Vec{v}_i[2]}\cdots a_m^{\Vec{v}_i[m]},
$
Parikh's Theorem is frequently formulated by giving the following its immediate consequence:

\begin{theorem}[\cite{Parikh1966}]\label{thm:Parikh}
	Every context-free language is Parikh equivalent to a regular language. 
\end{theorem}

Apparently, two \emph{unary languages} are Parikh equivalent if and only if they are equal.
Hence, as a consequence of Theorem~\ref{thm:Parikh}, each unary context-free language is regular.
This result, which was firstly discovered in~1962 by Ginsburg and Rice~\cite{GinsburgRice1962},
earlier than Parikh's Theorem, has been studied from the descriptional
complexity point of view in 2002 by Pighizzini, Shallit and Wang~\cite{PighizziniShallitWang2002}, 
proving the following:

\begin{theorem}[{\cite[Thms.~4, 6]{PighizziniShallitWang2002}}]\label{thm:PSW2002}
	For any \cnfg\ with $h$ variables that generates a unary language, there exist an 
	equivalent \ownfa\ with at most $2^{2h-1}+1$ states and an equivalent \owdfa\  
	with less than $2^{h^2}$ states. 
\end{theorem}

In the paper we will also make use of the transformation of unary \ownfas\ into~\owdfas, whose
cost was obtained in 1986 by Chrobak~\cite{Chrobak1986}:

\begin{theorem}[\cite{Chrobak1986}]\label{thm:Chrobak86}
	The state cost of the conversion of $n$-state unary \ownfas\ into equivalent \owdfas\ is $e^{\Theta(\sqrt{n \cdot \ln n})}$.
\end{theorem}

%---------------------------------------------------------------------
	\subsection{Preliminary Constructions}
	\label{sec:constructions}
%---------------------------------------------------------------------

Here, we present some preliminary constructions which will
be used in the rest of the paper. These constructions are simple and standard.
They are given just for the sake of completeness. Hence, the trained reader
can skip this part, going directly to the following sections.

First, we consider some \emph{decomposition results}: we show how to obtain automata and grammars 
for the unary and nonunary parts of the languages defined by given automata and grammars, respectively.
After, we will shortly  discuss some \emph{composition results}: 
how to obtain \owdfas\ or \twdfas, respectively,
accepting the union of languages defined by  given \owdfas\ or \twdfas.

Throughout the section, let us fix an $m$-symbol alphabet  $\Sigma=\{a_1, a_2,\linebreak \ldots,a_m\}$.
Let us starts by considering finite automata.

\begin{lemma}\label{lemma:decREG}
	For each $n$-state \ownfa~$A$ accepting a language $L(A)\subseteq\Sigma^*$,
	there exist $m+1$ \ownfas\ $A_0,A_1,\ldots,A_m$ such that:
	\begin{itemize}
	\item $A_0$ has $n(m+1)+1$ states and accepts the nonunary part of~$L(A)$.
	\item For $i=1,\ldots,m$, $A_i$ is a unary \ownfa\ with $n$ states which
	accepts the unary part $L(A)\cap\{a_i\}^*$.
	\end{itemize}
	Furthermore, if~$A$ is deterministic then $A_0,A_1,\ldots,A_m$ are deterministic, too.
\end{lemma}
\begin{proof}
  To accept an input $w$, the automaton~$A_0$ has to check that~$w$ is accepted by~$A$ and contains
	at least two different symbols. To do that, $A_0$ uses the same states and transitions as $A$.
	However, in a preliminary phase, it keeps track in its finite control of the first letter of~$w$,
	until discovering a different letter.
	
	The automaton~$A_0$, besides all states and transitions of~$A$,
	has a new initial state and $m$ extra copies $[q,1],\ldots,[q,m]$ of each state $q$ of $A$.
	The transitions from the initial state of~$A_0$ simulate those from the initial state of~$A$, also remembering the
	first symbol of the input, i.e., a transition in~$A$ which from the initial state, reading a symbol~$a_i$,
	leads to the state $q$, is simulated in~$A_0$ by a transition leading to $[q,i]$.
  
  From a state $[q,i]$, reading the same symbol~$a_i$ the automaton~$A_0$ can move to each state~$[p,i]$
  such that~$A$ from~$q$ reading~$a_i$ can move to~$p$. In this way, until the scanned input prefix consists
  only of occurrences of the same letter~$a_i$, $A_0$ simulates~$A$ using the $i$th copies of the states.
  However, when in a state $[q,i]$ a symbol $a_j\neq a_i$ is read, having verified that the input contains at 
  least two different letters, $A_0$ can move to each state~$p$ which is reachable in~$A$ from the state~$q$,
  so entering the part of~$A_0$ corresponding to the original~$A$.
  The final states of~$A_0$ are the final states in the copy of $A$.
  From this construction we see that the number of states of~$A_0$ is~$n(m+1)+1$.
  Furthermore, if~$A$ is deterministic then also~$A_0$ is deterministic.

\smallskip

	For the unary parts, it is easily seen that for $i=1,\ldots,m$, the automaton $A_i$ can be obtained 
	by removing from $A$ all the transitions on the symbols $a_j\neq a_i$. Clearly, also this construction preserves
	determinism.
%\qed
\end{proof}

We can give a similar result in the case of \cfgs.

\begin{lemma}\label{lemma:decCFL}
	For each $h$-variable \cnfg~$G$ generating a language $L(G)\subseteq\Sigma^*$,
	there exist $m+1$ \cnfgs\ $G_0,G_1,\ldots,G_m$ such that:
	\begin{itemize}
	\item $G_0$ has $mh{-}m{+}1$ variables and generates the nonunary part~$L_0$ of~$L(G)$.
	\item For $i=1,\ldots,m$, $G_i$ is a unary \cnfg\ with~$h$ variables which
	generates the unary part $L_i=L(G)\cap\{a_i\}^*$.
	\end{itemize}
\end{lemma}
\begin{proof}
For $i=1,\ldots,m$, the design of $G_i$ is simply done by deleting {}from $P$ all productions of the form $B \to a_j$ with $i \neq j$. 
Built in this manner, it is impossible for $G_i$ to contain more than $h$ variables. 

Giving such a linear upper bound on the number of variables for $G_0$ is slightly more involved. 
It is clear that any production of one letter or $\varepsilon$ directly {}from $S$ is contrary to the purpose of $G_0$. 
This observation enables us to focus on the derivations by $G$ that begins with replacing $S$ by two variables. 
Consider a derivation 
\[
S \derG BC \DerivesG u C \DerivesG uv
\] 
for some non-empty words  $u, v \in \Sigma^+$ and $S \to BC \in P$. 
$G_0$ simulates $G$, but also requires extra feature to test whether $u$ and $v$ contain respectively letters $a_i$ and $a_j$,
for some $i \neq j$, and make only derivations that pass this test valid.  
To this end, we let the start variable $S'$ of $G_0$ make guess which of the two distinct letters in 
$\Sigma$ have to derive {}from $B$ and $C$, respectively. 
We encode this guess into the variables in $V \setminus \{S\}$ as a subscript like $B_i$ (this means that, 
for $w \in \Sigma^*$, $B_i \DerivesZ w$ if and only if $B \DerivesG w$ and $w$ contains at least one $a_i$).

Now, we give a formal definition of $G_0$ as a quadruple $(V', \Sigma, P', S')$, where 
\[V' = \{S'\} \cup \{B_i \mid B \in V \setminus \{S\}, 1 \le i \le m\}\] 
and $P'$ consists of the following production rules: 
\begin{enumerate}
\item	$\{S' \to B_i C_j \mid \mbox{$S \to BC \in P$ and $1 \le i, j \le m$ with $i \neq j$}\}$; 
\item	$\{B_i \to C_i D_j, B_i \to C_j D_i \mid \mbox{$B \to CD \in\!P$ with $B \neq S$ and $1\!\le\!i, j\!\le\!m$}\}$; 
\item	$\{B_i \to a_i \mid \mbox{$B \to a_i \in P$ and $1 \le i \le m$}\}$. 
\end{enumerate}

\noindent
We conclude the proof by checking that $L(G_0)\!=\!\{w \in L(G)\!\mid\!\mbox{$w$ is not unary}\}$. 
To this aim we prove the following:
%It must be obvious for trained readers, and hence, they can skip the check and directly go after the end
%of the proof. 

\begin{claim}
	Let $B_i$ be a variable of $G_0$ that is different {}from the start variable. 
	For~$w \in \Sigma^*$, $B_i \DerivesZ w$ if and only if $B \DerivesG w$ and $w$ contains 
	at least one occurrence of~$a_i$. 
\end{claim}
%\begin{proof}[Proof of the claim]
\noindent
	Both implications will be proved using induction on the length of derivations. 

\smallskip

	\emph{(Only-if part):} 
	If $B_i \derZ w$ (single-step derivation), then $w$ must be $a_i$ and $B \to a_i \in P$ according to the type-3 production in $P'$. 
	Hence, the base case is correct. 
	The longer derivations must begin with either $B_i \to C_i D_j$ or $B_i \to C_j D_i$ for some $B \to CD \in P$ and some $1 \le j \le m$. 
	It is enough to investigate the former case (the other one is completely similar). 
	Then we have
	\[%
	B_i \derZ C_i D_j \DerivesZ w_1 D_j \DerivesZ w_1 w_2 = w%
	\]
  for some $w_1, w_2 \in \Sigma^+$. 
	By induction hypothesis, $C \DerivesG w_1$, $w_1$ contains $a_i$, and $D \DerivesG w_2$. 
	Hence, $B \derG CD \DerivesG w_1 D \DerivesG w_1 w_2 = w$ is a valid derivation by~$G$, and $w$ contains~$a_i$.

\smallskip

	\emph{(If part):} 
	The base case is proved as for the direct implication. 
	If $B \DerivesG w$ is not a single-step derivation, then it must start with applying to $B$ some production $B \to CD \in P$. 
	Namely,
	\[%
	B \derG CD \DerivesG w_1' D \DerivesG w_1' w_2' = w
	\] 
	for some non-empty words $w_1', w_2' \in \Sigma^+$. 
	Thus, either $w_1'$ or $w_2'$ contains $a_i$; let us say $w_1'$ does (the other case is similar).
	By induction hypothesis, $C_i \DerivesZ w_1'$. 
	A letter $a_j$ occurring in $w_2'$ is chosen, and the hypothesis gives $D_j \DerivesZ w_2'$. 
	As a result, the derivation 
	\[%
	B_i \derZ C_i D_j \DerivesZ w_1 D_j \DerivesZ w_1' w_2' = w
	\] 
	is valid. 
%\qed
%\end{proof}

This completes the proof of the claim.

\medskip

To conclude the proof of the lemma, let us check that $G_0$ genetates the nonunary part of $L(G)$. 
For the direct implication, assume that $u \in L(G_0)$. 
Its derivation should be 
\[%
S' \derZ B_i C_j \DerivesZ u_1 C_j \DerivesZ u_1 u_2 = u
\]
for some $S' \to B_i C_j \in P'$, $u_1, u_2 \in \Sigma^+$, $1\leq i,j\leq m$, with $i\neq j$. 
Ignoring the subscripts $i, j$ in this derivation brings us with $S \DerivesG u$. 
Moreover, the claim above implies that $u_1$ and $u_2$ contain $a_i$ and $a_j$, respectively. 
Thus, $u$ is a nonunary word in $L(G)$. 

Conversely, consider a nonunary word $w \in L(G)$. 
Being nonunary, $|w| \ge 2$, and this means that its derivation by $G$ must begin with a production $S \to BC$. 
Since $B, C \neq S$, they cannot produce $\varepsilon$, and hence, we have
\[%
S \derG BC \DerivesG w_1 C \DerivesG w_1 w_2 = w
\] 
for some nonempty words $w_1, w_2 \in \Sigma^+$. 
Again, being $w$ nonunary, we can find a letter $a_i$ in $w_1$ and a letter $a_j$ in $w_2$ such that $i \neq j$. 
Now, the claim above implies $B_i \DerivesZ w_1$ and $C_j \DerivesZ w_2$. 
Since $S' \to B_i C_j \in P'$, the derivation 
\[%
S' \derZ B_i C_j \DerivesZ w_1 C_j \DerivesZ w_1 w_2 = w
\]
is a valid one by $G_0$. 

Note that, being thus designed, $G_0$ contains $mh{-}m{+}1$ variables.
%\qed
\end{proof}

We conclude this section by shortly discussing some constructions related to the union of
languages defined by~\owdfas\ and by~\twdfas.

First, we remind the reader that  $k$ \owdfas~$A_1,A_2,\ldots,A_k$ with $n_1,n_2,\ldots,n_k$ states,
respectively, can be simulated ``in parallel'' by a \owdfa, in
order to recognize the union~$L(A_1)\cup L(A_2)\cup\cdots\cup L(A_k)$. In particular, the state set
of~$A$ is the cartesian product of the state sets of the given automata. For this reason,
the automaton~$A$ obtained according to this standard construction is
usually called \emph{product automaton}. Its number of states is~$n_1\cdot n_2\cdots n_k$.

\smallskip

If~$A_1,A_2,\ldots,A_k$ are \twdfas\ and we want to obtain a \twdfa~$A$
accepting the union~$L(A_1)\cup L(A_2)\cup\cdots\cup L(A_k)$,
the state cost reduces to the sum $n_1 + n_2 + \cdots + n_k$,
\emph{under the hypothesis that the automata are halting}, namely, they do not present any infinite computation.

In particular, on input~$w$, the automaton~$A$ simulates in sequence, for $i=1,\ldots,k$, the automata~$A_i$, 
halting and accepting in the case one $\hat\imath$ is found
such that $A_{\hat\imath}$ accepts $w$. 

Suppose that, for $i=1,\ldots,m$, the state set of $A_i$ is $Q_i$ with final state $q_{i,f}$
and $Q_i\cap Q_j\neq\emptyset$ for $i\neq j$.
Then, $A$ can be defined as follows.
\begin{itemize}
\item The set of states is $Q=Q_1\cup Q_2\cup\cdots\cup Q_k$.
\item The initial state is the initial state of $A_1$.
\item The final state is the final state $q_{k,f}$ of~$A_k$.
\item For $i=1,\ldots,k-1$, $A$ contains all the transitions of~$A_i$ with the exception
of those leading to the final state $q_{i,f}$ of $A_i$. Those transitions lead directly to~$q_{k,f}$, to halt and accept.

In this way, the state $q_{i,f}$ of $A_i$ becomes unreachable. (We remind the reader that this state is also
halting.) In the automaton~$A$, the state~$q_{i,f}$ is ``recycled'' in a different way: it is used to prepare the 
simulation of~$A_{i+1}$ after a not accepting
simulation of~$A_i$. To this aim, each undefined transition of~$A_i$ leads in~$A$ to the state~$q_{i,f}$,
where the automaton~$A$ loops, moving the head leftward, to reach the left endmarker. There, $A$~moves the head one
position to the right, on the first symbol of the input word, and 
enters the initial state of~$A_{i+1}$, hence starting to simulate it.
\item All the transitions of~$A_k$ are copied in~$A$ without any change. Hence, if the input was rejected in all the
simulations of $A_1,A_2,\ldots,A_{k-1}$, it is accepted by~$A$ if and only if it is accepted by~$A_k$.
\end{itemize}

We observe that each \owdfa\ can be converted into a \twdfa\ in the form we are considering (cf.\ p.~\pageref{p:twdfa}), 
just adding the accepting state, which is entered on the right endmarker when the given \owdfa\ accepts the input. 
So the above construction works (with the addition of at most $k$ extra states) 
even when some of the $A_i$'s are one-way.

Finally, we point out that, as proven in~\cite{GeffertMereghettiPighizzini2007},
with a linear increasing in the number of the states, each \twdfas\ can be made halting.
In particular, each $n$-state~\twdfa\ can be simulated by a halting~\twdfa\ with~$4n$ states.

So the above outlined construction can be extended to the case of nonhalting \twdfas\ by
obtaining a~\twdfa\ with no more than $4\cdot(n_1+n_2+\cdots+n_k)$ states.

%---------------------------------------------------------------------
%	\section{From one-way nondeterministic automata to Parikh equivalent one-way deterministic automata}
\section{From \mbox{\sc\ownfa}s to Parikh equivalent \mbox{\sc\owdfa}s}
	\label{sec:NFA_to_DFA}
%---------------------------------------------------------------------

In this section we present our first main contribution.
Fixed an alphabet~$\Sigma = \{a_1, a_2, \ldots, a_m\}$,
{}from each $n$-state \ownfa~$A$ with input alphabet~$\Sigma$, we derive a Parikh equivalent \owdfa\ $A'$ 
with $e^{O(\sqrt{n \cdot \ln n})}$
states. Furthermore, we prove that this cost is tight.

Actually, as a preliminary step, we obtain a result which is interesting \emph{per se}:
if each word accepted by the given \ownfa~$A$ contains at least two different symbols, i.e., it is
nonunary, then the Parikh equivalent \owdfa~$A'$ can be obtained with polynomially many states.
Hence, the superpolynomial blowup is due to the unary part of the accepted language.
This result (presented in Theorem~\ref{thm:nonunary_NFA_to_DFA}) looks quite
surprising. Hence, before starting the technical presentation,
we show an example with the aim to give, in a very simple case, a taste of our constructions.

\begin{example}\label{example}
Let us consider the following language
\[
L=\{ba^n\mid n\bmod 210\neq 0\}\,.
\]
Clearly, $L$ does not contain any unary word. Furthemore, it can be verified that~$L$ is accepted by the 
$18$-state \ownfa~$A$ in Fig.~\ref{fig:L}\emph{(Left)}.
In particular, in the initial state, reading the letter $b$, \emph{in a nondeterministic way} $A$ chooses 
to verify the membership of the input to one of the following languages:
\begin{itemize}
\item $L_1=\{ba^n\mid n\bmod 2\neq 0\}$\,,
\item $L_2=\{ba^n\mid n\bmod 3\neq 0\}$\,,
\item $L_3=\{ba^n\mid n\bmod 5\neq 0\}$\,,
\item $L_4=\{ba^n\mid n\bmod 7\neq 0\}$\,.
\end{itemize}
Of course, $L=L_1\cup L_2\cup L_3\cup L_4$.
The automaton~$A$ can be transformed into an equivalent \owdfa, by identifying the transitions leaving the initial
state and by merging the 4 loops into a unique loop of length $2\cdot 3\cdot 5\cdot 7 = 210$.
Using standard distinguishability arguments, it can be shown that it is not possible to do better. As a matter
of fact, the smallest complete \owdfa\ accepting~$L$ requires 212 states.

\begin{figure}[tb]
\begin{center}
\unitlength 0.225cm% = 5.691pt
\linethickness{0.4pt}%
\begin{picture}(40,26)
\put(0,0){%%%%% NFA
\begin{picture}(15,26)(-1,0)
\put(0,0){%
%7 state loop
\begin{picture}(14,6)(0,0)
  %states
  \put(1,3){\circle{2}}
  \put(5,5){\circle{2}}
  \put(9,5){\circle{2}}
  \put(13,5){\circle{2}}
  \put(13,1){\circle{2}}
  \put(9,1){\circle{2}}
  \put(5,1){\circle{2}}
  %final states
%  \put(1,3){\circle{1.6}}
  \put(5,5){\circle{1.6}}
  \put(9,5){\circle{1.6}}
  \put(13,5){\circle{1.6}}
  \put(13,1){\circle{1.6}}
  \put(9,1){\circle{1.6}}
  \put(5,1){\circle{1.6}}
  %transitions
  \put(1.75,3.75){\vector(2,1){2.25}}
  \put(6,5){\vector(1,0){2}}
  \put(10,5){\vector(1,0){2}}
  \put(13,4){\vector(0,-1){2}}
  \put(12,1){\vector(-1,0){2}}
  \put(8,1){\vector(-1,0){2}}
  \put(4,1.25){\vector(-2,1){2.25}}
  %labels
  \put(2.5,4.7){\makebox(0,0){\scriptsize $a$}}
  \put(6.7,5.5){\makebox(0,0){\scriptsize $a$}}
  \put(10.7,5.5){\makebox(0,0){\scriptsize $a$}}
  \put(13.5,3.5){\makebox(0,0){\scriptsize $a$}}
  \put(11.7,0.5){\makebox(0,0){\scriptsize $a$}}
  \put(7.7,0.5){\makebox(0,0){\scriptsize $a$}}
  \put(3.3,1){\makebox(0,0){\scriptsize $a$}}
\end{picture}
}%
\put(4,8){%
%5 state loop
\begin{picture}(10,6)(0,0)
  %states
  \put(1,3){\circle{2}}
  \put(5,5){\circle{2}}
  \put(9,5){\circle{2}}
  \put(9,1){\circle{2}}
  \put(5,1){\circle{2}}
  %final states
%  \put(1,3){\circle{1.6}}
  \put(5,5){\circle{1.6}}
  \put(9,5){\circle{1.6}}
  \put(9,1){\circle{1.6}}
  \put(5,1){\circle{1.6}}
  %transitions
  \put(1.75,3.75){\vector(2,1){2.25}}
  \put(6,5){\vector(1,0){2}}
  \put(9,4){\vector(0,-1){2}}
  \put(8,1){\vector(-1,0){2}}
  \put(4,1.25){\vector(-2,1){2.25}}
  %labels
  \put(2.5,4.7){\makebox(0,0){\scriptsize $a$}}
  \put(6.7,5.5){\makebox(0,0){\scriptsize $a$}}
  \put(9.5,3.5){\makebox(0,0){\scriptsize $a$}}
  \put(7.7,0.5){\makebox(0,0){\scriptsize $a$}}
  \put(3.3,1){\makebox(0,0){\scriptsize $a$}}
\end{picture}
}%
\put(8,16){%
%3 state loop
\begin{picture}(6,6)(0,0)
  %states
  \put(1,3){\circle{2}}
  \put(5,5){\circle{2}}
  \put(5,1){\circle{2}}
  %final states
%  \put(1,3){\circle{1.6}}
  \put(5,5){\circle{1.6}}
  \put(5,1){\circle{1.6}}
  %transitions
  \put(1.75,3.75){\vector(2,1){2.25}}
  \put(5,4){\vector(0,-1){2}}
  \put(4,1.25){\vector(-2,1){2.25}}
  %labels
  \put(2.5,4.7){\makebox(0,0){\scriptsize $a$}}
  \put(5.5,3.5){\makebox(0,0){\scriptsize $a$}}
  \put(3.3,1){\makebox(0,0){\scriptsize $a$}}
\end{picture}
}
\put(8,24){%
%2 state loop
\begin{picture}(6,2)(0,0)
  %states
  \put(1,1){\circle{2}}
  \put(5,1){\circle{2}}
  %final state
  \put(5,1){\circle{1.6}}  
  %transitions
  \put(1.85,1.5){\vector(1,0){2.3}}
  \put(4.15,0.5){\vector(-1,0){2.3}}
  %labels
  \put(2.7,2){\makebox(0,0){\scriptsize $a$}}
  \put(3.3,0){\makebox(0,0){\scriptsize $a$}}
\end{picture}
}
%
%%initial state
\put(-1,25){\vector(1,0){1}}
%\put(-1,26){\vector(1,-1){1}}
\put(1,25){\circle{2}}
%%transitions from initial state
\put(2,25){\vector(1,0){6}}
\put(1.75,24.30){\vector(3,-2){6.7}}
\put(1.45,24.1){\vector(1,-4){3.05}}
\put(1,24){\vector(0,-1){20}}
%labels
\put(3.5,25.7){\makebox(0,0){\scriptsize $b$}}
\put(3.5,23.8){\makebox(0,0){\scriptsize $b$}}
\put(2.5,21.7){\makebox(0,0){\scriptsize $b$}}
\put(0.5,21.7){\makebox(0,0){\scriptsize $b$}}

\end{picture}
}
\put(25,0){%%%% PARIKH EQUIVALENT DFA
\begin{picture}(15,26)(-1,0)
\put(0,0){%
%7 state loop
\begin{picture}(14,6)(0,0)
  %states
  \put(1,3){\circle{2}}
  \put(5,5){\circle{2}}
  \put(9,5){\circle{2}}
  \put(13,5){\circle{2}}
  \put(13,1){\circle{2}}
  \put(9,1){\circle{2}}
  \put(5,1){\circle{2}}
  %final states
  \put(1,3){\circle{1.6}}
  \put(5,5){\circle{1.6}}
  \put(9,5){\circle{1.6}}
  \put(13,5){\circle{1.6}}
%  \put(13,1){\circle{1.6}}
  \put(9,1){\circle{1.6}}
  \put(5,1){\circle{1.6}}
  %transitions
  \put(1.75,3.75){\vector(2,1){2.25}}
  \put(6,5){\vector(1,0){2}}
  \put(10,5){\vector(1,0){2}}
  \put(13,4){\vector(0,-1){2}}
  \put(12,1){\vector(-1,0){2}}
  \put(8,1){\vector(-1,0){2}}
  \put(4,1.25){\vector(-2,1){2.25}}
  %labels
  \put(2.5,4.7){\makebox(0,0){\scriptsize $a$}}
  \put(6.7,5.5){\makebox(0,0){\scriptsize $a$}}
  \put(10.7,5.5){\makebox(0,0){\scriptsize $a$}}
  \put(13.5,3.5){\makebox(0,0){\scriptsize $a$}}
  \put(11.7,0.5){\makebox(0,0){\scriptsize $a$}}
  \put(7.7,0.5){\makebox(0,0){\scriptsize $a$}}
  \put(3.3,1){\makebox(0,0){\scriptsize $a$}}
\end{picture}
}%
\put(4,8){%
%5 state loop
\begin{picture}(10,6)(0,0)
  %states
  \put(1,3){\circle{2}}
  \put(5,5){\circle{2}}
  \put(9,5){\circle{2}}
  \put(9,1){\circle{2}}
  \put(5,1){\circle{2}}
  %final states
  \put(1,3){\circle{1.6}}
  \put(5,5){\circle{1.6}}
  \put(9,5){\circle{1.6}}
%  \put(9,1){\circle{1.6}}
  \put(5,1){\circle{1.6}}
  %transitions
  \put(1.75,3.75){\vector(2,1){2.25}}
  \put(6,5){\vector(1,0){2}}
  \put(9,4){\vector(0,-1){2}}
  \put(8,1){\vector(-1,0){2}}
  \put(4,1.25){\vector(-2,1){2.25}}
  %labels
  \put(2.5,4.7){\makebox(0,0){\scriptsize $a$}}
  \put(6.7,5.5){\makebox(0,0){\scriptsize $a$}}
  \put(9.5,3.5){\makebox(0,0){\scriptsize $a$}}
  \put(7.7,0.5){\makebox(0,0){\scriptsize $a$}}
  \put(3.3,1){\makebox(0,0){\scriptsize $a$}}
\end{picture}
}%
\put(8,16){%
%3 state loop
\begin{picture}(6,6)(0,0)
  %states
  \put(1,3){\circle{2}}
  \put(5,5){\circle{2}}
  \put(5,1){\circle{2}}
  %final states
  \put(1,3){\circle{1.6}}
  \put(5,5){\circle{1.6}}
%  \put(5,1){\circle{1.6}}
  %transitions
  \put(1.75,3.75){\vector(2,1){2.25}}
  \put(5,4){\vector(0,-1){2}}
  \put(4,1.25){\vector(-2,1){2.25}}
  %labels
  \put(2.5,4.7){\makebox(0,0){\scriptsize $a$}}
  \put(5.5,3.5){\makebox(0,0){\scriptsize $a$}}
  \put(3.3,1){\makebox(0,0){\scriptsize $a$}}
\end{picture}
}
\put(8,24){%
%2 state loop
\begin{picture}(6,2)(0,0)
  %states
  \put(1,1){\circle{2}}
  \put(5,1){\circle{2}}
  %final state
  \put(5,1){\circle{1.6}}  
  %transitions
  \put(1.85,1.5){\vector(1,0){2.3}}
  \put(4.15,0.5){\vector(-1,0){2.3}}
  %labels
  \put(2.7,2){\makebox(0,0){\scriptsize $a$}}
  \put(3.3,0){\makebox(0,0){\scriptsize $a$}}
\end{picture}
}
%
%%initial state
\put(-1,25){\vector(1,0){1}}
\put(1,25){\circle{2}}
%%transitions from initial state
\put(2,25){\vector(1,0){6}}
\put(1,24){\vector(0,-1){4}}
%labels
\put(3.5,25.7){\makebox(0,0){\scriptsize $b$}}
\put(0.5,23){\makebox(0,0){\scriptsize $a$}}

%%2nd state from top
\put(1,19){\circle{2}}
%%transitions 
\put(2,19){\vector(1,0){6}}
\put(1,18){\vector(0,-1){4}}
%labels
\put(3.5,19.7){\makebox(0,0){\scriptsize $b$}}
\put(0.5,17){\makebox(0,0){\scriptsize $a$}}

%%3rd state from top
\put(1,13){\circle{2}}
%%transitions 
\put(2,12.75){\vector(2,-1){2.25}}
\put(1,12){\vector(0,-1){4}}
%labels
\put(3.5,12.7){\makebox(0,0){\scriptsize $b$}}
\put(0.5,11){\makebox(0,0){\scriptsize $a$}}

%%4th state from top
\put(1,7){\circle{2}}
%%transitions 
\put(1,6){\vector(0,-1){2}}
%labels
\put(0.5,5){\makebox(0,0){\scriptsize $b$}}

\end{picture}
}
\end{picture}
\end{center}
\caption{\emph{(Left)} The \ownfa~$A$ of Example~\ref{example}.  \emph{(Right)} The Parikh equivalent \owdfa~$A'$.}
\label{fig:L}
\end{figure}
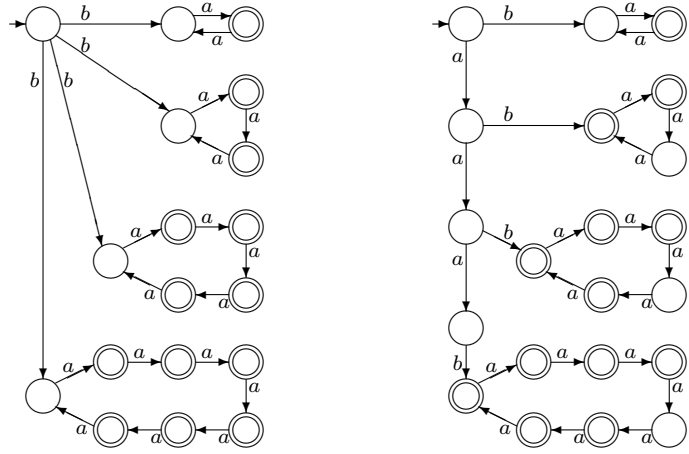

However, we can build a complete \owdfa~$A'$ with only $22$ states, accepting a language $L'$ Parikh equivalent to $L$.
To do that, for $i=1,\ldots,4$, we replace each language $L_i$ with a Parikh equivalent language $L'_i$ in such a way 
that all the words in~$L'_i$ begin with the prefix $a^{i-1}b$, and then we define $L'$ as the union of the resulting 
languages, namely:
\begin{itemize}
\item $L'_1=\{ba^n\mid n\bmod 2\neq 0\}=L_1$\,,
\item $L'_2=\{aba^{n-1}\mid n\bmod 3\neq 0\}$\,,
\item $L'_3=\{a^2ba^{n-2}\mid n\bmod 5\neq 0\}$\,,
\item $L'_4=\{a^3ba^{n-3}\mid n\bmod 7\neq 0\}$\,,
\item $L'=L'_1\cup L'_2\cup L'_3\cup L'_4$\,.
\end{itemize}
In this way, given an input word $w$, after reading the first 4 input symbols, \emph{in a deterministic way}~$A'$ 
can decide to which language $L'_i$, $1\leq i\leq 4$, test the membership
of $w$, in order to decide whether or not $w\in L'$.

The automaton~$A'$ is depicted in Fig.~\ref{fig:L}\emph{(Right)}. The vertical path starting from the initial
state is used to select, depending on the position of the letter~$b$, one of the loops, i.e., to select which
language $L'_i$ must be used to decide the membership of the input to $L'$. (Of course, when the symbol $b$ does not
appear in the prefix of length $4$, the automaton rejects by entering a dead state, which is not depicted.)

The loops of~$A'$ are obtained by suitably ``unrolling'' the loops of the original \ownfa~$A$. The unrolled
parts of the loops are moved before $b$-transitions and merged together in the vertical path which starts from the
initial state.\qed
\end{example}

\begin{lemma}\label{lem:representation}
  There exists a polynomial~$p$ such that for each $n$-state \ownfa\  $A$ over $\Sigma$,
  the Parikh image of the language accepted by~$A$ can be written as
	\begin{equation}\label{eq:nf}
	\psi(L(A)) = Y\cup\bigcup_{i \in I} Z_i\,,
	\end{equation}
	where:
	\begin{itemize}
	\item $Y\subseteq\mathbb{N}^m$ is a finite set of vectors whose components are bounded by $p(n)$;
	\item $I$ is a set of at most $p(n)$ indices;
	\item for each $i\in I$, $Z_i \subseteq \mathbb{N}^m$ is a linear set of the form: 
	\begin{equation}\label{eq:Zi}
		Z_i = \{\Vec{v}_{i, 0} + n_1 \Vec{v}_{i, 1} + n_2 \Vec{v}_{i, 2} + \cdots + n_{k_i} \Vec{v}_{i, k_i}\mid n_1,n_2, \ldots, n_{k_i} \in \mathbb{N}\}\,,
	\end{equation}
	with:
	\begin{itemize}
	\item $0\le k_i\le m$,
	\item the components of the offset $\Vec{v}_{i, 0}$ are bounded by $p(n)$, 
	\item the generators $\Vec{v}_{i, 1},\Vec{v}_{i, 2}, \ldots, \Vec{v}_{i, k_i}$ are linearly independent 
	vectors from $\{0, 1, \ldots, n\}^m$.
	\end{itemize}
	\end{itemize}
	Futhermore, if all the words in $L(A)$ are nonunary then
	for each $i \in I$ we can choose a {\it nonunary} vector $\Vec{x}_i \in \Pred(\Vec{v}_{i, 0})$ such that all
	those chosen vectors are pairwise distinct. 
\end{lemma}
\begin{proof}
  In~\cite[Thms.\ 7 and 8]{KopczynskiTo2010} it was proved that $\psi(L(A))$ can be written as claimed in the
  first part of the statement of the lemma, with $Y=\emptyset$, $I$ of size polynomial in $n$,
  and the components of each offset $\Vec{v}_{i, 0}$ bounded by~$O(n^{2m}m^{m/2})$.\footnote{%
	For the sake of completeness, we derive a rough upper bound for the cardinality of the set $I$ by counting
	the number of possible combinations of offsets and generators satisfying the given limitations.
	Since the components of the offsets are bounded by~$O(n^{2m}m^{m/2})$, the number of possible
	different offsets is~$O(n^{2m^2}m^{m^2\!/2})$.
	Furthermore, there are $(n+1)^m$ vectors in~$\{0, 1, \ldots, n\}^m$. Hence, $n^{m^2}$ is an upper bound
	for the number of possible sets of~$k$ generators, with $k=1,\ldots,m$.
	This allows us to give~$O(n^{3m^2}m^{m^2\!/2})$ as an upper bound for the cardinality of~$I$.
	We point out that in~\cite[Thm.\ 4.1]{To2010} slightly different bounds have been done.
	However, nothing is said about linearly independency of generators.\label{noteA}
}%end-footnote

	Now we prove the second part of the statement. Hence, let us suppose that all the words in~$L(A)$ are
	nonunary. Notice that this implies that also all the offsets~$\Vec{v}_{i, 0}$ are nonunary.

  If for each $i\in I$ we can choose $\Vec{x}_i \in \Pred(\Vec{v}_{i, 0})$ such that all
  $\Vec{x}_i$'s are pairwise different, then the proof is completed.
  
  Otherwise, we proceed as follows.
  For a vector $\Vec{v}$, let us denote by $\|\Vec{v}\|$ its \emph{infinite norm}, i.e., the value
  of its  maximum component.
  Let us suppose $I\subseteq\mathbb{N}$ and denote as $N_I$ the maximum element of $I$.
  
  By proceeding in increasing order, for $i\in I$ we choose a nonunary vector 
  $\Vec{x}_i \in \Pred(\Vec{v}_{i, 0})$
  such that $\|\Vec{x}_i\|\le i$ and $\Vec{x}_i$ is different from all already chosen $\Vec{x}_j$,
  i.e., $\Vec{x}_i\neq\Vec{x}_j$ for all $j\in I$ with $j<i$. The extra condition $\|\Vec{x}_i\|\le i$ 
  will turn out to be useful later.
  
  When for an $i\in I$ it is not possible to find such $\Vec{x}_i$,
  we replace $Z_i$ by some suitable sets. Essentially, those sets are obtained by
  enlarging the offsets using sufficiently long ``unrollings'' of the generators.
  In particular, for $j=1,\ldots,k_i$, we consider the set
	\begin{equation}\label{eq:newZi}
		Z_{N_I+j} = \{(\Vec{v}_{i, 0}+h_j\Vec{v}_{i, j}) + n_1 \Vec{v}_{i, 1} + \cdots + n_{k_i} \Vec{v}_{i, k_i}\mid n_1, \ldots, n_{k_i} \in \mathbb{N}\}\,,
	\end{equation}
	where $h_j$ is an integer satisfying the inequalities
	\begin{equation}\label{eq:hj}
		N_I+j\le \|\Vec{v}_{i, 0}+h_j\Vec{v}_{i, j}\|<N_I+j+n
	\end{equation}
	Due to the fact that $\Vec{v}_{i, j}\in\{0,\ldots,n\}^m$, we can always find such $h_j$.
	Furthermore, we consider the following finite set
	\begin{equation}\label{eq:Yi}
		Y_i = \{\Vec{v}_{i, 0} + n_1 \Vec{v}_{i, 1} + \cdots + n_{k_i} \Vec{v}_{i, k_i}\mid 0\le n_1<h_1, \ldots, 0\le n_{k_i}<h_{k_i}\}\,.
	\end{equation}
  It can be easily verified that
  \[
  Z_i=Y_i\cup\bigcup_{j=1}^{k_i}Z_{N_I+j}\,.
  \]  
  Now we replace the set of indices $I$ by the set
  \[
  \widehat I=I-\{i\}\cup\{N_I+1,\ldots,N_I+k_i\}\,,
  \]
  and the set $Y$ by $\widehat Y=Y\cup Y_i$. We continue the same process by considering the next index $i$.
  
  We notice that, since we are choosing each vector $\Vec{x}_i\in\Pred(\Vec{v}_{i, 0})$ in such a way that
  $\|\Vec{x}_i\|\le i$, when we will have to choose the vector $\Vec{x}_{N_I+j}$ for a set
  $Z_{N_I+j}$ introduced at this stage, by the condition (\ref{eq:hj})
  we will have at least one possibility (a vector with one component equal to $N_I+j$ and another component
  equal to~$1$; we remind the reader that, since the given automaton accepts only nonunary words, 
  all offsets are nonunary).
  This implies that after examining all sets $Z_i$ corresponding to the original set $I$, we do not need
  to further modify the sets introduced during this process. Hence, this procedure ends in a finite number of steps.
  
  Furthermore, for each $Z_i$ in the initial representation, we introduced at most $m$
  sets. Hence, the cardinality $\widetilde N$ of the set of indices resulting at the end of this process 
  is still polynomial.\footnote{%
  Since the cardinality of the set of indices, before the transformation, was~$O(n^{3m^2}m^{m^2\!/2})$ 
  (cf. Note~\ref{noteA}),
  the cardinality~$\widetilde N$ after the transformation is $O(n^{3m^2}m^{m^2\!/2+1})$.\label{noteB}%
}%end-footnote
  
  By (\ref{eq:hj}) the components of the offsets which have been added in this process cannot
  exceed $\widetilde N + n$. Hence, it turns out that $m\cdot(\widetilde N+n)$ is an upper
  bound to the components of vectors in $Y_i$. This permit to conclude that $p(n)=m\cdot(\widetilde N+n)$
  is an upper bound for all these amounts.\footnote{%
  Hence $p(n)=O(n^{3m^2}m^{m^2/2+2})$.\label{noteC}%
}%end-footnote    
%\qed
\end{proof}

Now we are able to consider the case of automata accepting only words that are nonunary.

\begin{theorem}\label{thm:nonunary_NFA_to_DFA}
	For each $n$-state \ownfa\ over~$\Sigma$, accepting a language none of whose words are unary, there exists a Parikh equivalent \owdfa\ with a number   
	of states polynomial in $n$. 
\end{theorem}
\begin{proof}
Let $A$ be the given $n$-state \ownfa. According to Lemma~\ref{lem:representation}, we express the Parikh image
of $L(A)$ as in~(\ref{eq:nf}) and, starting {}from this representation, we will build a \owdfa\ $A_{\rm non}$ 
that is Parikh equivalent to $A$.
To this end, we could apply the following procedure:
\begin{enumerate}
	\item\label{step:dfa1}For each $i\in I$, build a \owdfa\ $A_i$ such that $\psi(L(A_i))=Z_i$.
	\item\label{step:dfa2}{}From the automata $A_i$'s so obtained, derive a \owdfa\ $A'$ such that $\psi(L(A'))=\bigcup_{i \in I} Z_i$.
	\item\label{step:dfa3}Define a \owdfa\ $A''$ such that $\psi(L(A''))=Y$.
	\item\label{step:dfa4}From $A'$ and $A''$, using the standard construction for the union, build a~\owdfa~$A_{\rm non}$ such 
	that~$L(A_{\rm non})=L(A')\cup L(A'')$ and, hence, $\psi(L(A_{\rm non}))=Y\cup\bigcup_{i \in I} Z_i$,
	i.e., $A_{\rm non}$ is Parikh equivalent to~$A$.
\end{enumerate}
Actually, we will use a variation of this procedure. In particular, steps~\ref{step:dfa1} and~\ref{step:dfa2}, to
obtain~$A'$, are modified as we now explain.

\begin{figure}[tb]
\begin{center}
\includegraphics[scale=0.7]{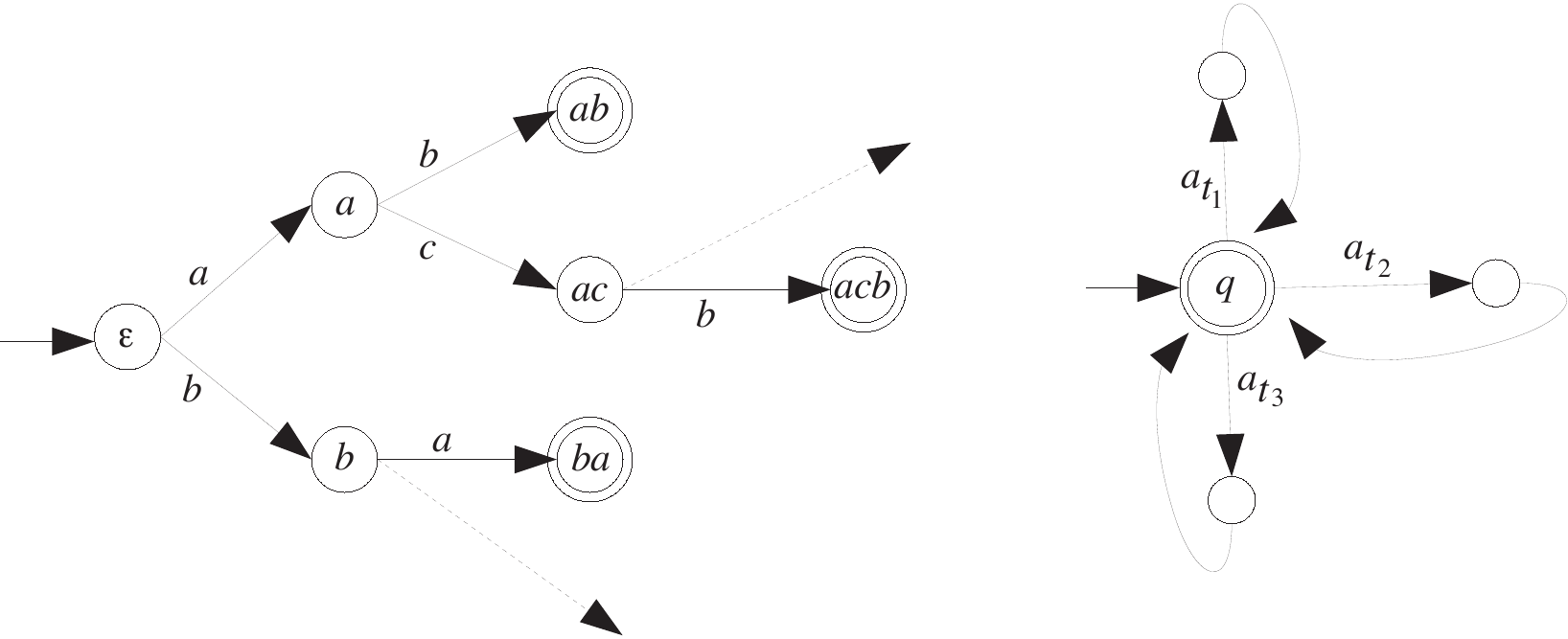}
\end{center}
\caption{
	\emph{(Left)} A construction of \owdfa\ $A_W$, where the state $q_u$ is simply denoted by $u$ for clarity. 
	\emph{(Right)} A construction of \owdfa\ $B_i$ that accepts $\{w_{i, 1}, w_{i, 2}, \ldots, w_{i, k_i}\}^*$ for $k_i= 3$.
	In the construction of the final \owdfa\ $A_{\rm non}$, if $w_{i, 0} = acb$, then the initial state $q$ of $B_i$ is merged with the state $q_{acb}$ of~$A_W$. 
} 
\label{fig:DFA}
\end{figure}

\smallskip

Let us start by considering $i\in I$.
First, we handle the generators of $Z_i$. To this aim,
let us consider the function $g: \mathbb{N}^m \to \Sigma^*$ defined by
\[
g(\Vec{v}) = a_1^{i_1} a_2^{i_2} \cdots a_m^{i_m}\,,
\]
for each vector $\Vec{v} = (i_1, \ldots, i_m)\in\mathbb{N}^m$.

Using this function, we map the generators $\Vec{v}_{i, 1}, \ldots, \Vec{v}_{i, k_i}$ into the words 
\[
s_{i, 1}= g(\Vec{v}_{i, 1}),\  s_{i, 2}= g(\Vec{v}_{i, 2}),\ \ldots,\ s_{i, k_i}= g(\Vec{v}_{i, k_i})\,.
\] 
It is easy to define an automaton accepting the language $\{s_{i, 1}, s_{i, 2}, \ldots, s_{i, k_i}\}^*$, which consists of a start state $q$ with $k_i$ loops labeled with $s_{i, 1}, s_{i, 2}, \ldots, s_{i, k_i}$, respectively. 
The state $q$ is the only  accepting state. 
However, this automaton is nondeterministic. 

To avoid this problem, we modify the language by replacing each $s_{i, j}$, for $j = 1, \ldots, k_i$, with a Parikh equivalent word $w_{i, j}$ in such a way that for all pairwise different $j,j'$ the
corresponding words $w_{i, j}$ and $w_{i, j'}$ begin with different letters. 
This is possible due to the fact, being $\Vec{v}_{i, 1}, \ldots, \Vec{v}_{i, k_i}$ linearly independent,
according to Lemma~\ref{lem:indep} we can find $k_i$ different letters 
$a_{t_1}, a_{t_2}, \ldots, a_{t_{k_i}} \in \Sigma$ such that $\Vec{v}_{i, j}[t_j] > 0$ for $j = 1, \ldots, k_i$.
We ``rotate'' each $s_{i, j}$ by a cyclic shift so that the resulting word, $w_{i, j}$, begins with an occurrence of the letter $a_{t_j}$. 
Then $w_{i, j}$ is Parikh equivalent to $s_{i, j}$. 
For example, if $s_{i, j} = a_1^3 a_2^4 a_3$ and $t_j = 2$, then $w_{i, j}$ should be chosen as $a_2^4 a_3 a_1^3$. 

The construction of a \owdfa\ $B_i$ with one unique accepting state $q$ that accepts $\{w_{i, 1}, w_{i, 2}, \ldots, w_{i, k_i}\}^*$ must be now clear: $q$ with $k_i$ loops labeled with these respective $k_i$ words 
(see Fig.~\ref{fig:DFA}\emph{(Right)}). 
Furthermore, due to the limitations deriving {}from Lemma~\ref{lem:representation}, 
the length of these loops is at most $mn$ so that this \owdfa\ contains at most $1+m(mn-1)$ states. 

\smallskip
Now, we can modify this \owdfa\ in order to build an automaton~$A_i$ recognizing a language whose Parikh image is $Z_i$.
To this aim, it is enough to add a path which from an initial state, reading a word $w_{i, 0}$ with Parikh
image~$\Vec{v}_{i, 0}$, reaches the state~$q$ and then the part accepting $\{w_{i, 1}, w_{i, 2}, \ldots, w_{i, k_i}\}^*$
that we already described.
Due to the limitations on $\Vec{v}_{i, 0}$, this can be done by adding a polynomial number of states.
In particular, we could take $w_{i, 0}=g(\Vec{v}_{i, 0})$, thus completing step~\ref{step:dfa1}.
However, when we have such \owdfas\ for all $i\in I$, by applying the standard construction for the union to them, 
as in step~\ref{step:dfa2}, being $I$
polynomial in~$n$, the resulting \owdfa\ could have exponentially many states in~$n$, namely it could be too large
for our purposes.

To avoid this problem, the automaton~$A'$ which should be derived from steps~\ref{step:dfa1}-\ref{step:dfa2},
is obtained by using a different strategy.
We introduce the function $f: \mathbb{N}^m \to \Sigma^*$ defined as: for $\Vec{v} \in \mathbb{N}^m$, $f(\Vec{v}) = \mbox{$\hookleftarrow$}(g(\Vec{v}))$, where $\hookleftarrow$ denotes the 1-step left circular shift. 
For example, $f(4, 1, 2, 0, \ldots, 0) = a_1^3 a_2 a_3^2 a_1$.
It can be verified that the 1-step left circular shift endows $f$ with the prefix property {\it over the nonunary vectors}, that is, for any $\Vec{u}, \Vec{v} \in \mathbb{N}^m$ that are nonunary, if $f(\Vec{u})$ is a prefix of $f(\Vec{v})$, then $\Vec{u} = \Vec{v}$. 
Let 
\[
w_{i, 0} = f(\Vec{x}_i) g(\Vec{v}_{i, 0} - \Vec{x}_i)\,,
\] 
where $\Vec{x}_i\in\Pred(\Vec{v}_{i, 0})$ is given by Lemma~\ref{lem:representation}.
Clearly, $\psi(w_{i, 0})=\Vec{v}_{i, 0}$.
We now consider the finite language
\[
W = \{w_{i, 0} \mid i \in I\}\,.
\]
Because the $\Vec{x}_{i}$'s are nonunary and $f$ has the prefix property over nonunary words,
the language $W$ is prefix-free. 
We build a (partial) \owdfa\ that accepts~$W$\!, which is denoted by $A_W = (Q_W, \Sigma, q_\varepsilon, \delta_W, F_W)$, where:
\begin{itemize}
\item $Q_W = \{q_u \mid u \in \Pref(W)\}$,
\item the state $q_\varepsilon$ corresponding to the empty word is the initial state,
\item $F_W = \{q_u \mid u \in W\}$,
\item $\delta_W$ is defined as: for $u \in \Pref(W)$ and $a \in \Sigma$, if $ua \in \Pref(W)$, then $\delta(q_u, a) = q_{ua}$, while $\delta(q_u, a)$ is undefined otherwise.
\end{itemize}
See Fig.~\ref{fig:DFA}\emph{(Left)}. 
Clearly, this accepts $W$\!. 
Since the longest word(s) in $W$ is of length $m\cdot p(n)$, this \owdfa\ contains at most 
$1+|I| \cdot m \cdot p(n)=O(mp^2(n))$ states.

It goes without saying that each accepting state of this \owdfa\ is only for one word in $W$\!,
namely, two distinct words in~$W$ are accepted by~$A_W$ at distinct two states. 
Now, based on $A_W$ and the \owdfas\ $B_i$ with $i \in I$, we can build a finite automaton that accepts the language $\bigcup_{i \in I} w_{i, 0} L(B_i)$ {\it without introducing any new state}. 
This is simply done by merging $q_{w_{i, 0}}$ with the start state of $B_i$. 
Given an input $u$, the resulting automaton $A'$ simulates the \owdfa\ $A_W$, looking for a prefix $w$ of $u$ such that $w \in W$\!. 
When such a prefix is found, $A'$ starts simulating $B_i$ on the remaining suffix $z$, where $i$ is the index such that $w = w_i$. 
Since $W$ is prefix-free, we need only to consider one decomposition of the input as $u = wz$. 
This implies that $A'$ is deterministic. 
Finally, we observe that
$A'$ contains at most $O(mp^2(n))+ |I| (1+ m(mn-1))=O(mp(n)(p(n)+mn))$ states, i.e.,
a number which is polynomial in $n$.

\smallskip

We now sketch the construction of a \owdfa\ $A''$ accepting a language $L_Y$ whose Parikh image is $Y$
(step~\ref{step:dfa3}).
We just take $L_Y=\{g(\Vec{v})\mid\Vec{v}\in Y\}$. Let $M$ be the maximum of the components
of vectors in $Y$. With each $\Vec{v}\in\{0,\ldots,M\}^m$, we associate a state
$q_{\Vec{v}}$ which is reachable from the initial state by reading the word
$g(\Vec{v})$. Final states are those corresponding to vectors in $Y$.
The automaton $A''$ so obtained has $(M+1)^m=(p(n)+1)^m$ states, a number polynomial in $n$.

Finally, by applying the standard construction for the union (step~\ref{step:dfa4}), from automata $A'$ and $A''$
we obtain the \owdfa~$A_{\rm non}$ Parikh equivalent to the given \ownfa\ $A$,
with number of states polynomial in~$n$.\footnote{%
Assuming $p(n)\ge mn$, the number of states of~$A'$ is~$O(mp^2(n))$.
Hence, the number of states of~$A_{\rm non}$ is~$O(mp^{m+2}(n))$.
Using~$p(n)=O(n^{3m^2}m^{m^2/2+2})$~(cf.\ Note~\ref{noteC}), we conclude that the number of state
of~$A_{\rm non}$ is~$O(mn^{3m^2(m+2)}m^{(m^2/2+2)(m+2)}) = O(n^{3m^3+6m^2}m^{m^3/2+m^2+2m+5})$.
Hence, it is polynomial in the number of states of the original \ownfa~$A$,
but exponential in the alphabet size~$m$.\label{noteD}
}
%\qed
\end{proof}

We now switch to the  general case. We prove that for each input alphabet
the state cost of the conversion of \ownfas\
into Parikh equivalent \owdfas\ is the same as for the unary alphabet.

\begin{theorem}\label{thm:NFA_to_DFA}
	For each $n$-state \ownfa\ over~$\Sigma$, there exists a Parikh equivalent \owdfa\ with $e^{O(\sqrt{n \cdot \ln n})}$ states. 
	Furthermore, this cost is tight. 
\end{theorem}
\begin{proof}
According to Lemma~\ref{lemma:decREG}, {}from a given $n$-state \ownfa\ $A$ with input alphabet 
$\Sigma = \{a_1, a_2, \ldots, a_m\}$, we build a \ownfa\ $A_0$ with $n(m+1)+1$ states that accepts
the nonunary part of~$L(A)$ and $m$ $n$-state \ownfas\ $A_1,A_2,\ldots,A_m$ that accept the unary parts
of~$L(A)$.
Using Theorem~\ref{thm:Chrobak86}, for $i=1,\ldots,m$, we convert $A_i$ into an equivalent 
\owdfa\ $A_i'$ with $e^{O(\sqrt{n \cdot \ln n})}$ states. We can assume that the state sets of the resulting automata are
  pairwise disjoint.
  
  We define $A_u$ that accepts $\{w \in L(A) \mid \mbox{$w$ is unary}\}$
  consisting of one copy of each of these \owdfas\ and a new state $q_s$, which is its start state.
  In reading the first letter $a_i$ of an input, $A_u$ transits {}from $q_s$ to the state $q$ in the 
  copy of $A'_i$ if $A'_i$ transits {}from its start state to $q$ on $a_i$ (such $q$ is unique because $A'_i$ 
  is deterministic). These transitions {}from $q_s$ do not introduce any nondeterminism because 
  $A'_1, \ldots, A'_m$ are defined over pairwise distinct letters. 
	After thus entering the copy, $A_u$ merely simulates $A'_i$. 
	The start state $q_s$ should be also an accepting state if and only if $\varepsilon \in L(A'_i)$ for some 
	$1 \le i \le m$.  
Being thus built, $A_u$ accepts exactly all the unary words in~$L(A)$ and contains at most 
$m \cdot e^{O(\sqrt{n \cdot \ln n})}+1$ states.
 
On the other hand, for the nonunary part of~$L(A)$, using Theorem~\ref{thm:nonunary_NFA_to_DFA},  we convert~$A_0$
into a Parikh equivalent \owdfa\ $A_n$ with a number of states $r(n)$, polynomial in $n$. 
The standard product construction is applied to $A_u$ and $A_n$ in order to build a \owdfa\ accepting $L(A_u) \cup L(A_n)$. 
The number of states of the \owdfa\ thus obtained is bounded by the product 
$e^{O(\sqrt{n \cdot \ln n})}\cdot r(n) = e^{O(\sqrt{n \cdot \ln n})}\cdot e^{O(\ln n)} = e^{O(\sqrt{n \cdot \ln n}+\ln n)}$, 
which is still bounded by $e^{O(\sqrt{n \cdot \ln n})}$.

Finally, we observe that by Theorem~\ref{thm:Chrobak86}, in the unary case $e^{\Theta(\sqrt{n \cdot \ln n})}$ is the
tight cost of the conversion {}from $n$-state \ownfas\ to \owdfas. This implies that the upper bound we obtained here
cannot be reduced.

This completes the proof of Theorem~\ref{thm:NFA_to_DFA}. 
%\qed
\end{proof}

We conclude this section with some observations.
We proved that fixed an alphabet~$\Sigma$, the state cost of the conversion of $n$-state
\ownfas\ into Parikh equivalent \owdfas\ is polynomial in~$n$, in the case each word in the accepted
language is nonunary. Otherwise, the cost is exponential in~$\sqrt{n \cdot \ln n}$.
A closer inspection to our proofs shows that these costs are exponential in the size of the
alphabet (see Notes~\ref{noteA}-\ref{noteD}).

While the cost of the conversion in the general case has been proved to be tight (see Theorem~\ref{thm:NFA_to_DFA}),
it should be interesting to see whether or not the  cost for the conversion in the case $n$-state
\ownfas\ accepting only nonunary words (Theorem~\ref{thm:nonunary_NFA_to_DFA}) could be further reduced.
To this respect, we point out that~$n$ is a lower bound. In fact, a smaller cost would imply that any given \ownfa\ (or \owdfa)
$B_0$, could be converted into a smaller Parikh equivalent \owdfa~$B_1$ which, in turn, could be further
converted in a smaller Parikh equivalent \owdfa~$B_2$ an so on.
In this way, from~$B_0$ we could build an arbitrary long sequence of automata $B_0,B_1,B_2,\ldots$,
all of them Parikh equivalent to~$B_0$, and such that for each $i>0$, $B_i$ would be smaller than~$B_{i-1}$.
This clearly does not make sense. With a similar argument, we can also conclude that even the costs of
the conversion of  $n$-state \ownfas\ into a Parikh equivalent \ownfas\ must be at least~$n$.

%----------------------------------------------------------------
	\section{From \mbox{\sc\cfg}s to Parikh equivalent \mbox{\sc\owdfa}s}
	\label{sec:CFG_to_DFA}
%----------------------------------------------------------------

In this section we extend the results of Section~\ref{sec:NFA_to_DFA} to
the conversion of \cfgs\ in Chomsky normal form into Parikh equivalent \owdfas.
Actually, Theorem~\ref{thm:nonunary_NFA_to_DFA} will play an important role in order
to obtain the main result of this section. The other important ingredient is the
following result proven in~2011 by Esparza, Ganty, Kiefer and Luttenberger~\cite{EsGaKiLu2011}, 
which gives the cost of the conversion of \cnfgs\ into Parikh equivalent \ownfas.

\begin{theorem}[\cite{EsGaKiLu2011}]\label{thm:EGKL2011}
	For each \cnfg\ with $h$ variables, there exists a Parikh equivalent \ownfa\ 
	with~${2h+1\choose h}=O(4^h)$ states. 
\end{theorem}

We point out that the upper bound in Theorem~\ref{thm:EGKL2011} does not depend on
the cardinality of the input alphabet.

By combining Theorem~\ref{thm:EGKL2011} with the main result of the previous section,
i.e., Theorem~\ref{thm:NFA_to_DFA}, we can immediately obtain a double exponential
upper bound in~$h$ for the size of \owdfas\ Parikh equivalent to \cnfgs\ with $h$ variables.
However, we can do better. In fact, we show how to reduce the upper bound to a single exponential 
in a polynomial of $h$. We obtain this result by proceeding as in the case of finite automata:
we split the language defined by given grammar into the unary and nonunary parts, we make separate conversions, 
and finally we combine the results.

\smallskip

As in Section~\ref{sec:NFA_to_DFA}, from now on let us fix the alphabet $\Sigma=\{a_1, a_2, \ldots, a_m\}$.
We start by considering the nonunary part.
By combining Theorem~\ref{thm:EGKL2011} with Theorem~\ref{thm:nonunary_NFA_to_DFA} we obtain:

\begin{theorem}\label{thm:nonunary_CFG_to_DFA}
	For each $h$-variable \cnfg\ with terminal alphabet~$\Sigma$, generating a language none of whose words are unary, 
	there exists a Parikh equivalent \owdfa\ with~$2^{O(h)}$ states. 
\end{theorem}
\begin{proof}
  First, according to Theorem~\ref{thm:EGKL2011}, we can transform the grammar into a Parikh equivalent
  \ownfa\ with~$O(4^h)$ states. Then, using Theorem~\ref{thm:nonunary_NFA_to_DFA}, we convert the resulting automaton
  into a Parikh equivalent \owdfa, with a number of states polynomial in $4^h=2^{2h}$, hence exponential
  in~$h$.
\end{proof}

Now, we switch to the general case.

\begin{theorem}\label{thm:CNFG_to_DFA}
	For each $h$-variable \cnfg\ with terminal alphabet~$\Sigma$, there exists a Parikh 
	equivalent \owdfa\ with at most~$2^{O(h^2)}$ states. 
\end{theorem}
\begin{proof}
Let us denote the given \cnfg\ by $G = (V, \Sigma, P, S)$, where $|V| = h$. 

In the case $m = 1$ (unary alphabet), one can employ Theorem~\ref{thm:PSW2002} (note that, over a unary alphabet, two languages $L_1, L_2$ are Parikh equivalent if and only if they are equivalent). 
Hence, {}from now on we assume $m\geq 2$. 

Let us give an outline of our construction first: 
\begin{enumerate}
\item\label{step:split}
	{}From $G$, we first create \cnfgs\ $G_0, G_1, \ldots, G_m$ such that $G_0$ generates the nonunary part of $L(G)$ and
	$G_1,G_2,\ldots,G_m$ generate the unary parts.
\item\label{step:unary}	
	The grammars $G_1, G_2, \ldots, G_m$ are converted into respectively equivalent unary \owdfas\ $A_1, A_2, \ldots, A_m$.
	{}From these \owdfas, a \owdfa\ $A_{\rm unary}$ accepting the set of all unary
	words in $L(G)$ is constructed. 
\item\label{step:CNFG_to_DFA}
	The grammar $G_0$ is converted into a Parikh equivalent \owdfa~$A_{\rm non}$.
\item\label{step:last_merge}
	Finally, {}from $A_{\rm unary}$ and $A_{\rm non}$, a \owdfa\ that accepts the union of $L(A_{\rm unary})$ and $L(A_{\rm non})$
	is obtained. 
\end{enumerate}
Observe that $L(A_{\rm unary}) = \{w \in L(G) \mid \mbox{$w$ is unary}\}$ and $L(A_{\rm non})$ is Parikh equivalent to $L(G_0) = \{w \in L(G) \mid \mbox{$w$ is not unary}\}$. 
Thus, the \owdfa\ which is finally constructed by this procedure is Parikh equivalent to the given grammar~$G$.

We already have all the tools we need to implement each step in the above construction.
\begin{itemize}  
\item[\ref{step:split}.]
  We can obtain grammars $G_0, G_1, \ldots, G_m$ according to Lemma~\ref{lemma:decCFL}.
  In particular, $G_0$ has $mh-m+1$ variables, while each of the remaining grammars has $h$ variables.
\item[\ref{step:unary}.]
  According to Theorem~\ref{thm:PSW2002}, for $i=1,\ldots,m$, grammar $G_i$ is converted into a \owdfa\ $A_i$
  with less than $2^{h^2}$ states. Using the same strategy presented in the proof of Theorem~\ref{thm:NFA_to_DFA},
  {}from $A_1,\ldots,A_m$, we define $A_{\rm unary}$ consisting of one copy of each of these \owdfas\ and a 
  new state $q_s$, which is its start state. Hence, the number of states of $A_{\rm unary}$ does not
  exceed $m 2^{h^2}$. 
\item[\ref{step:CNFG_to_DFA}.]
  This step is done using Theorem~\ref{thm:nonunary_CFG_to_DFA}. The number of the states of the resulting
  \owdfa~$A_{\rm non}$ is exponential in the number of the variables of
  the grammar $G_0$ and, hence, exponential in $h$.
\item[\ref{step:last_merge}.]
  The final \owdfa\ can be obtained as the product of two automata $A_{\rm unary}$ and $A_{\rm non}$.
  Considering the bounds obtained in Step~\ref{step:unary} and~\ref{step:CNFG_to_DFA} we conclude that
  the number of states in exponential in $h^2$.\footnote{%
  We briefly discuss how the upper bounds depends on~$m$, the alphabet size.
  Using the estimation of the cost of the conversions in Theorem~\ref{thm:nonunary_NFA_to_DFA}
  (see Note~\ref{noteD}) and
  observing that the grammar~$G_0$ in the previous construction has less than $mh$ variables,
  we can conclude that the automaton~$A_{\rm non}$ has~$2^{O(hm^4)}$ states.
  Hence, the number of the states of the \owdfa\ finally obtained by the construction given in the proof
  of Theorem~\ref{thm:CNFG_to_DFA} is~$2^{O(h^2+hm^4)}$.\label{noteE}
}%
\qedhere
\end{itemize}
%\qed
\end{proof}

\noindent
We point out that in~\cite{LavadoPighizzini2012} it has been proved a result close to Theorem~\ref{thm:CNFG_to_DFA}
in the case of \cnfgs\ generating \emph{letter bounded languages}, i.e., subsets of $a_1^*a_2^*\cdots a_m^*$.
In particular, an upper bound exponential in a polynomial in $h$ has been obtained.
However, the degree of the polynomial is, in turn, a polynomial in the size $m$ of the alphabet.
Here, in  our Theorem~\ref{thm:CNFG_to_DFA}, the degree is~$2$. Hence, it does not depend on~$m$.

\smallskip
We observe that in~\cite[Thm.\ 7]{PighizziniShallitWang2002} it was proved that there is a constant
$c>0$ such that for infinitely many $h>0$ there exists a \cnfg\ with $h$ variables generating a unary
language such that each equivalent \owdfa\ requires at least $2^{ch^2}$ states. This implies that the upper bound
in Theorem~\ref{thm:CNFG_to_DFA} cannot be improved.

\smallskip

Even the bound given in Theorem~\ref{thm:nonunary_CFG_to_DFA}, for languages consisting only
of nonunary words, cannot be improved, by replacing the exponential in~$h$ by a slowly increasing
function. 
This can be shown by adapting a standard argument from the unary case 
(e.g.,~\cite[Thm.~5]{PighizziniShallitWang2002}).
For any integer $h\ge 3$, consider the grammar~$G$ with variables $A,B,A_0,A_1,\ldots,A_{h-3}$,
and productions
\[
A\rightarrow a\,,~ B\rightarrow b\,,~ A_0\rightarrow AB\,,~ A_j\rightarrow A_{j-1}A_{j-1}\,,\mbox { for $j=1,...,h-3$}
\,.
\]
As easy induction shows that, for $j=1,\ldots,h-3$, the only word which is generated from~$A_j$
is $(ab)^{2^j}$. 
Hence, by choosing $A_{h-3}$ as start symbol, we have~$L(G)=\{(ab)^H\}$, with~$H=2^{h-3}$. 
An immediate pumping argument shows that each \owdfa\ (or even \ownfa) with
less than~$2H+1$ states accepting a word of length~$2H$, should also  accept some
words of length $<2H$. Since $L(G)$ contains only the word~$(ab)^H$, it turns out that each 
\owdfa\ accepting a language Parikh equivalent to $L(G)$ requires
$2H+1$ states, namely a number exponential in~$h$.

%----------------------------------------------------------------
	\section{Conversions into Parikh equivalent \mbox{\sc\twdfa}s}
	\label{sec:to_2DFA}
%----------------------------------------------------------------

In this section we study the conversions of \ownfas\ and \cnfgs\ into Parikh equivalent
\emph{two-way} deterministic automata.
In the previous sections, for the conversions into \emph{one-way} deterministic automata,
we observed that the unary parts are the most expensive. However, the cost
of the conversions of unary \ownfas\ and \cnfgs\ into \twdfas\ are smaller than the costs for
the corresponding conversions into \owdfas. 
This allows us to prove that, in the general case, the cost of the conversions
of  \ownfas\ and \cnfgs\ into Parikh equivalent \twdfas\ are smaller
than the cost of the corresponding conversions into \owdfas.

Let us start by presenting the following result, which derives from~\cite{Chrobak1986}:

\begin{theorem}\label{thm:u1NFAto2DFA}
  For each $n$-state unary \ownfa\ there exists an equivalent halting \twdfa\ with $n^2+1$ states.
\end{theorem}
\begin{proof}
  For the sake of completeness, we present a proof which is essentially the same
  given by Chrobak~\cite[Thm.~6.2]{Chrobak1986} where, however,	the obtained upper bound was~$O(n^2)$. 
  Then we will explain why the big-$O$ in the upper bound can be removed.
	
	First of all, each $n$-state unary \ownfa~$A$ can be converted into an equivalent
	$\ownfa$~$A_c$ in a special form, which is known as \emph{Chrobak normal form}~\cite[Lemma~4.3]{Chrobak1986}, 
	consisting of a deterministic path which starts from the initial state,
	and $k\geq 0$ disjoint deterministic cycles. The number of states in the path is~$s=O(n^2)$,
	while the \emph{total} number of states in the cycles is $r\leq n$.
	From the last state of the path there are $k$ outgoing edges, each one of them reaching
	a fixed state on a different cycle. Hence, on each input of length $\ell\geq s$, the computation visits
	all the states on the initial path, until the last one where the \emph{only} nondeterministic
	choice is taken, moving to one of the cycles, where the remaining part of the input
	is examined. However, if $\ell<s$ then computation ends in a state on the initial path, without reaching
	any loop. (In the special case $k=0$ the accepted language is finite.)
	
	A \twdfa~$B$ can simulate the $\ownfa$~$A_c$ in Chrobak normal form, traversing the input
	word at most $k+1$ times. In the first traversal, the automaton checks whether or not the input
	length is $<s$. If this is the case, then the automaton accepts or rejects according to the
	corresponding state on the initial path of~$A_c$. Otherwise, it moves to the right endmarker.
	This part can be implemented with $s+1$ states ($s$ states for the simulation of the initial path,
	plus one more state to move to the right endmarker).
	From the right endmarker, the automaton traverses the input leftward, by simulating the first
	cycle of $A_c$ from a suitable state (which is fixed, only depending on $s$ and on the cycle length).
	If the left endmarker is reached in a state which simulates a final state in the cycle then the
	automaton~$B$ moves to the final state $q_f$ and accepts, otherwise it
	traverses the input  rightward, simulating the $2$nd cycle of~$A_c$, and so on.
	Hence, in the $(i+1)$th traversal of the input, $1\leq i\leq k$, the $i$th cycle is simulated.
	So the number of states used to simulate the cycles is equal to the total number of
	states in the cycles, namely~$r$.
	Considering the final state $q_f$, we conclude that $B$ can be implemented with $s+r+2=O(n^2)$ states.
	
	Finally, we point out that finer estimations for the number of the states on the initial path and in the
	loops of~$A_c$ have been found. In~\cite{Geffert2007}, it was proved that the number of $s$ of the states
	in the initial path is bounded by $n^2-2$ and the sum $r$ of the numbers of the states in the cycles 
	is bounded by $n-1$.\footnote{%
	Actually, there is an exception: if the given \ownfa~$A$ is just one cycle of $n$ states then~$A$ is
	already a \owdfa. If it is minimal, then in any equivalent \ownfa\ we cannot have a cycle with less than~$n$ states
	which is useful to accept some input.
	However, in this degenerate case, Theorem~\ref{thm:u1NFAto2DFA} is trivially true, without making use of the Chrobak
	normal form.%
	}
	The first bound has been further reduced in~\cite{Gawrychowski2011}
	to $s\leq n^2-n$. This allows us to conclude that the \twdfa~$B$ can be obtained
	with at most $n^2+1$ states.
%\qed
\end{proof}

The upper bound given in Theorem~\ref{thm:u1NFAto2DFA} is asymptotically tight. As proven in~\cite[Thm.~6.3]{Chrobak1986},
for each integer $n$ there exists an $n$-state unary \ownfa\ such that any equivalent \twdfa\ requires
$\Omega(n^2)$ states.

\medskip

By combining Theorem~\ref{thm:u1NFAto2DFA} with the bound for the transformation of
unary \cnfgs\ into \ownfas\ given in Theorem~\ref{thm:PSW2002}, we immediately obtain the
following bound.

\begin{theorem}\label{thm:uCFGto2DFA}
  For each  $h$-variable unary \cnfg\ there exists an equivalent halting \twdfa\ with 
  at most~$(2^{2h-1}+1)^2+1$ states.
\end{theorem}

We now have the tools for studying the conversions of \ownfas\ and \cfgs\ into Parikh equivalent \twdfas.
Let us start with the first conversion.

\begin{theorem}\label{thm:1NFAto2DFA}
  For each $n$-state \ownfa\ there exists a Parikh equivalent \twdfa\ with a number of states polynomial in $n$.
\end{theorem}
\begin{proof}
  We use the same technique as in the proof of Theorem~\ref{thm:NFA_to_DFA}, by splitting the language accepted
  by the given \ownfa\ $A$ into its unary and nonunary parts, as explained in Lemma~\ref{lemma:decREG}. 
  Each unary part is accepted by a \ownfa\ with 
  $n$ states. According to Theorem~\ref{thm:u1NFAto2DFA}, this gives us $m$ \twdfas\ $B_1,B_2,\ldots,B_m$,
  accepting the unary parts, each one
  of them has at most~$n^2+1$ states, where $m$ is the cardinality of the input alphabet
  $\Sigma=\{a_1,a_2,\ldots, a_m\}$.
    
  For the nonunary part we have a \ownfa\ with $n(m+1)+1$ states and, according to Theorem~\ref{thm:nonunary_NFA_to_DFA},
  a Parikh equivalent \owdfa\ $B_0$ with a number of states polynomial in~$n(m+1)+1$ and, hence, in~$n$.
  
  Finally, as explained in Section~\ref{sec:constructions}, we can build a \twdfa\ $B$ such that 
  $L(B)=L(B_0)\cup L(B_1)\cup\cdots\cup L(B_m)$.
  Hence, $B$ is Parikh equivalent to the given \ownfa~$A$ and its number of states is polynomial in~$n$.\footnote{%
  By making the same considerations as in Notes~\ref{noteC} and~\ref{noteD}, we can obtain an~$O(m^3)$
  bound for the degree of the polynomial.%
}
%\qed
\end{proof}

Now, we consider the conversion of~\cfgs.

\begin{theorem}\label{thm:CFGto2DFA}
  For each $h$-variable \cnfg\ there exists a Parikh equivalent \twdfa\ with $2^{O(h)}$ states.
\end{theorem}
\begin{proof}
  Even in this case, the construction is obtained by adapting the corresponding conversion
  into \owdfas\ (Theorem~\ref{thm:CNFG_to_DFA}).
  In particular, the construction uses the same steps~\ref{step:split}-\ref{step:last_merge} given
  in that proof, with some modifications in steps~\ref{step:unary} and~\ref{step:last_merge}, which are replaced
  by the following ones:
\begin{itemize}
  \item[\ref{step:unary}'.]
	The grammars $G_1, G_2, \ldots, G_m$ are converted into respectively equivalent unary \twdfas\ $A'_1, A'_2, \ldots, A'_m$.
\item[\ref{step:last_merge}'.]
	Finally, {}from $A_{\rm non}, A'_1, A'_2, \ldots, A'_m$, a \twdfa\ that accepting the language
	$L(A_{\rm non})\cup L(A'_1)\cup L(A'_2)\cup\cdots\cup L(A'_m)$ is obtained.
\end{itemize}
  Clearly, the \twdfas\ resulting from this procedure is Parikh equivalent to the original grammar~$G$.
  The costs of steps~\ref{step:split} and~\ref{step:CNFG_to_DFA} has been discussed in the proof
  of Theorem~\ref{thm:CNFG_to_DFA}. For the remaining steps:
\begin{itemize}
  \item[\ref{step:unary}'.]
	According to Theorem~\ref{thm:uCFGto2DFA}, for $i=1,\ldots,m$, the \twdfa~$A'_i$ has at most
	$(2^{2h-1}+1)^2+1$ states.
\item[\ref{step:last_merge}'.]
	We use the construction presented at the end of Section~\ref{sec:preliminaries}, to obtain a
	\twdfa\ whose number of states is the sum of the number of the states of $A_{\rm non}, A'_1, A'_2, \ldots, A'_m$,
	hence $2^{O(h)}$.\footnote{Explicitly mentioning the dependency on the alphabet size~$m$, we can give 
	a~$2^{O(hm^4)}$ bound. This derives from the size of the automaton~$A_{\rm non}$ (cf. Note~\ref{noteE}).
}
\qedhere
\end{itemize}
%\qed
\end{proof}

%---------------------------------------------------------------------
	\section{Conclusion}
	\label{sec:conclu}
%---------------------------------------------------------------------

We proved that the state cost of the conversion of $n$-state \ownfas\ into Parikh equivalent \owdfas\ is 
$e^{\Theta(\sqrt{n \cdot \ln n})}$. This is the same cost of the conversion of unary \ownfas\
into equivalent \owdfas.
Since in the unary case Parikh equivalence is just equivalence, this result can be seen as
a generalization of the Chrobak conversion~\cite{Chrobak1986} to the nonunary case.
More surprisingly, such a cost is due to the unary parts of the languages. In fact, as shown
in Theorem~\ref{thm:nonunary_NFA_to_DFA}, for each $n$-state unary \ownfa\ accepting a language
which does not contain any unary word there exists a Parikh equivalent \owdfa\ with polynomially many states.
Hence, while for the transformation {}from \ownfas\ to equivalent \owdfas\ we need at least two different
symbols to prove the exponential gap {}from $n$ to $2^n$ states and we have a smaller gap
in the unary case, for Parikh equivalence the worst case is due only to unary words.

Even in the proof of our result for \cfgs\ (Theorem~\ref{thm:CNFG_to_DFA}), the separation between the
unary and nonunary parts was crucial. Also in this case, it turns out that the most expensive
part is the unary one.

On the other hand, in our conversions into Parikh equivalent \twdfas, the most expensive part turns out
to be the nonunary one.

\paragraph*{Acknowledgment}
The authors wish to thank Viliam Geffert who suggested to study the conversions into
Parikh equivalent \emph{two-way} automata. In particular, Theorem~\ref{thm:1NFAto2DFA}
is due to him. Many thanks also to Jeff Shallit who suggested to add Example~\ref{example}.

%---------------------------------------------------------------------
%\bibliographystyle{model1-num-names}
\bibliographystyle{alpha}%%
	\bibliography{pcomp}
%---------------------------------------------------------------------

\end{document}